\documentclass[11pt]{article}

\usepackage[margin=1in]{geometry}
\usepackage{graphicx} \usepackage{amsmath,amsfonts,amssymb,amsthm,thmtools,bm,dsfont} 

\usepackage{hyperref}
\hypersetup{
	colorlinks=true,
	linkcolor=blue,
	citecolor=blue,
	filecolor=blue,
	urlcolor=blue
}
\usepackage[nameinlink]{cleveref}
\usepackage{enumitem}
\setitemize[1]{leftmargin=5mm}
\usepackage{xcolor}
\usepackage{booktabs}
\usepackage{caption}
\usepackage{subcaption}
\usepackage{xspace}
\usepackage[round,authoryear]{natbib}

\newtheorem{theorem}{Theorem}
\newtheorem{conjecture}[theorem]{Conjecture}
\theoremstyle{definition}

\newtheorem{lemma}[theorem]{Lemma}

\crefname{lemma}{lemma}{lemmas}
\Crefname{lemma}{Lemma}{Lemmas}
\crefname{theorem}{theorem}{theorems}
\Crefname{theorem}{Theorem}{Theorems}

\def\ddefloop#1{\ifx\ddefloop#1\else\ddef{#1}\expandafter\ddefloop\fi}
\def\ddef#1{\expandafter\def\csname #1\endcsname{\ensuremath{\mathbb{#1}}}}
\ddefloop ABCDFGHIJKLMNOQRTUVWXYZ\ddefloop  \def\ddef#1{\expandafter\def\csname c#1\endcsname{\ensuremath{\mathcal{#1}}}}
\ddefloop ABCDEFGHIJKLMNOPQRSTUVWXYZ\ddefloop  \def\ddef#1{\expandafter\def\csname b#1\endcsname{\ensuremath{\bm #1}}}
\ddefloop ABCDEFGHIJKLMNOPQRSTUVWXYZabcdeghijklnopqrstuvwxyz\ddefloop  

\DeclareMathOperator*{\Ex}{\mathbb{E}}

\newcommand{\DSBS}{\mathsf{DSBS}}

\newcommand{\SBES}{\mathsf{SBES}}
\newcommand{\SBESp}{\mathsf{SBES}_p}

\newcommand{\sign}{\mathsf{sign}}
\newcommand{\oS}{\overline{S}}
\newcommand{\bone}{\mathds{1}}

\newcommand{\Stab}{\mathsf{Stab}}
\newcommand{\Maj}{\mathsf{Maj}}
\newcommand{\Dict}{\mathsf{Dict}}

\renewcommand{\lor}{\mbox{ or }}
\renewcommand{\land}{\mbox{ and }}
\newcommand{\Bin}{\mathrm{Bin}}

\newcommand{\Ber}{\mathrm{Ber}}
\newcommand{\Paren}[1]{\left(#1\right)}
\newcommand{\Brack}[1]{\left[#1\right]}
\newcommand{\Abs}[1]{\left|#1\right|}

\newcommand{\counterf}{g_n}
\newcommand{\difff}{\delta_n}
\newcommand{\eps}{\varepsilon}
\newcommand{\bzero}{\mathbf{0}}

\renewcommand{\setminus}{\smallsetminus}
\newcommand{\bhole}{e}

\newcommand{\diffset}{E}

\newcommand{\mlsc}{\textrm{MLS}\xspace}
\newcommand{\nicdc}{\textrm{NICD-E}\xspace}

\allowdisplaybreaks

\title{
When Majority Fails:\\
Tight Bounds for Correlation Distillation Conjectures
}
\author{
    Pritish Kamath \\
    \texttt{pritish@alum.mit.edu}
    \and
    Ravi Kumar\\
    \texttt{ravi.k53@gmail.com}\\[3mm]
    Google Research
    \and
    Pasin Manurangsi\\
    \texttt{pasin@google.com}
}
\date{}

\begin{document}

\maketitle

\begin{abstract}
We study two conjectures posed in the analysis of Boolean functions $f : \{-1, 1\}^n \to \{-1, 1\}$, in both of which, the {\em Majority} function plays a central role: the ``Majority is Least Stable''~\citep{benjamini1999noise} and the ``Non-Interactive Correlation Distillation for Erasures''~\citep{Yan03,odonnell12new}.

While both conjectures have been refuted in their originally stated form, we obtain a nearly tight characterization of the noise parameter regime in which each of the conjectures hold, for all $n \ge 5$. Whereas, for $n=3$, both conjectures hold in all noise parameter regimes. We state refined versions of both conjectures that we believe captures the spirit of the original conjectures.

\end{abstract}

\section{Introduction}

The {\em Non-Interactive Correlation Distillation (NICD)} problem, parameterized by a joint distribution $P$ over $\cX \times \cY$, is to maximize for $(\bx, \by) \sim P^{\otimes n}$, the correlation $\Ex [f(\bx) \cdot g(\by)]$, subject to $\Ex[f(\bx)] = \Ex[g(\by)] = 0$ for functions $f : \cX^n \to [-1, 1]$ and $g : \cY^n \to [-1, 1]$. Let $\rho^\star_n(P)$ denote this optimal value, and let $\rho^\star(P) := \sup_n \rho^\star_n(P)$.
Informally, this corresponds to a two-player setting, where Alice receives a sequence $\bx = (x_1, \ldots, x_n) \in \cX^n$, Bob receives a sequence $\by = (y_1, \ldots, y_n) \in \cY^n$, where each $(x_i, y_i) \overset{\text{i.i.d.}}{\sim} P$. Each player is required to return a single bit, that is marginally uniform, while maximizing the probability that their returned bits agree.\footnote{The range of $f$ and $g$ is $[-1, 1]$, which can be interpreted as Alice returning a randomized rounding of $f(\bx)$ to $\{-1, 1\}$, and similarly Bob returning a randomized rounding of $g(\by)$.}

Specifically, we consider NICD for the following distributions (defined in \Cref{tab:dsbs-sbes}):
\begin{enumerate}
\item the {\em Doubly Symmetric Binary Source} ($\DSBS$) defined over $\{-1, 1\} \times \{-1, 1\}$, and
\item the {\em Symmetric Binary Erasure Source} ($\SBES$) defined over $\{-1, 1\} \times \{-1, 1, \star\}$.
\end{enumerate}

\subsection{\boldmath \texorpdfstring{$\DSBS$}{DSBS} and the ``Majority is Least Stable (\mlsc)'' Conjecture}
The {\em Doubly Symmetric Binary Source} $\DSBS_\rho$, parameterized by a given $\rho \in [0, 1]$, is a joint distribution $(x, y)$ over $\{-1, 1\} \times \{-1, 1\}$ given in \Cref{tab:dsbs-sbes}(a), where the marginal distributions are uniform, and the correlation between corresponding bits is $\Ex[xy] = \rho$.
It is folklore (see \Cref{sec:related}) that the dictator functions, namely $f(\bx) = \Dict(\bx) := x_1$ and $g(\by) = \Dict(\by) := y_1$, realize the optimal value of $\rho^\star(\DSBS_\rho)$ for all values of $\rho \in [0, 1]$.

\begin{table}[t]
\renewcommand{\arraystretch}{1.5}
  \centering
\begin{minipage}{0.3\textwidth}
    \centering
    \begin{tabular}{r|cc}
      \hline
      & $Y=1$ & $Y=-1$ \\
      \hline
      $X=1$ & $\frac{1+\rho}{4}$ & $\frac{1-\rho}{4}$ \\
      $X=-1$ & $\frac{1-\rho}{4}$ & $\frac{1+\rho}{4}$ \\
      \hline
    \end{tabular}
    \caption*{(a) $\text{DSBS}_\rho$}
  \end{minipage}
  \hspace{3mm} \begin{minipage}{0.45\textwidth}
    \centering
    \begin{tabular}{r|ccc}
      \hline
      & $Y=1$ & $Y=\star$ & $Y=-1$ \\
      \hline
      $X=1$ & $\frac{p}{2}$ & $\frac{1-p}{2}$ & $0$ \\
      $X=-1$ & $0$ & $\frac{1-p}{2}$ & $\frac{p}{2}$ \\
      \hline
    \end{tabular}
    \caption*{(b) $\text{SBES}_p$}
  \end{minipage}
  \caption{(a) Doubly Symmetric Binary Source and (b) Symmetric Binary Erasure Source.}
  \label{tab:dsbs-sbes}
\end{table}

The special case when $f = g$ is often studied through the notion of {\em noise stability}, $\Stab_{\rho}(f) := \Ex_{(\bx, \by)} [f(\bx) f(\by)]$ for $(\bx, \by) \sim \DSBS_\rho^{\otimes n}$~\citep[see][]{ryanbook}. So, in other words, the dictator function is the most noise stable.

\citet{benjamini1999noise} studied the noise stability of {\em linear threshold functions} (LTF), namely, functions of the form $f(\bx) = \sign(\sum_{i=1}^n w_i x_i)$, and they posed the question of which (unbiased) LTF has the least noise stability. They conjectured that for odd $n$, $\Maj_n(\bx) := \sign(\sum_{i=1}^n x_i)$ has the least noise stability among unbiased LTFs, for any $\rho \in (0, 1)$. This \emph{\mlsc} conjecture was also highlighted by \citet{filmus2014real}.

However, as stated, the \mlsc conjecture was shown to be false. \citet{gopi2013stability} obtained a counterexample for all odd $n \ge 5$, which~\citet{BiswasS23} later proved to be minimal by showing that the conjecture holds for $n \le 3$.\footnote{As mentioned by \citet{BiswasS23}, the counterexample for $n=5$ were also discovered by Noam Berger, Daniel Kane, Steven Heilman,  and~\citet{jain2017counterexample}.}
Gopi also showed a partial positive result: for any fixed $n$, the conjecture holds in a small neighborhood of $\rho$ around $1$.

\subsection{\boldmath \texorpdfstring{$\SBES$}{SBES} and the ``NICD for Erasures (\nicdc)'' Conjecture}

\citet{Yan03} studied the NICD problem for the {\em symmetric binary erasure source}. Parameterized by $p \in [0, 1]$, $\SBESp$ is a joint distribution $(x, y)$ over $\{-1, 1\} \times \{-1, 1, \star\}$ given in \Cref{tab:dsbs-sbes}(b), where the marginal distribution of $x$ is uniform, and $y$ is obtained by ``erasing'' $x$ with probability $1 - p$.

It is easy to see that for any choice of $f : \cX^n \to [-1, 1]$, the optimal $g^\star : \cY^n \to [-1, 1]$ that maximizes $\Ex [f(\bx) \cdot g(\by)]$ is given by $g^\star(\by) := \sign(\Ex [f(\bx) | \by])$. And thus, maximizing $\Ex [f(\bx) \cdot g(\by)]$ becomes equivalent to maximizing $\Phi_p(f)$, where
$$\Phi_p(f) := \Ex_{\by} \Brack{\big | \Ex [f(\bx) \mid  \by] \big |} \qquad \text{subject to} \qquad \Ex [f(\bx)] = 0.$$
Note that it does not follow that $\Ex[g^\star(\by)] = 0$. Following prior literature, we ignore this constraint and instead only require that $\Ex [f(\bx)] = 0$. Although, note that for {\em odd} $f$, namely when $f(\bx) = -f(-\bx)$ holds for all $\bx$, it follows that $g^\star$ is also unbiased.

For the low-erasure regime ($p > 1/2$),~\citet{odonnell12new} showed that the dictator function $\Dict(\bx) := x_1$ maximizes $\Phi_p(f)$. Furthermore, for the high-erasure regime ($p < 1/2$)
and odd $n$, they conjectured that the Majority function $\Maj_n$ maximizes $\Phi_p(f)$; this \emph{\nicdc} conjecture was also highlighted by \citet{filmus2014real}.

Recently,~\citet{ivanisvili2025nicd}  presented a counterexample\footnote{Interestingly, they discovered the counterexample with help of a large language model. While their stated counterexample $\sign(x_1 - 3x_2 + x_3 - x_4 + 3 x_5)$ appears highly asymmetric, it is equivalent (up to permutations and sign flips) to the counterexample discovered for the \mlsc conjecture, namely $\sign(2x_1 + 2x_2 + x_3 + x_4 + x_5)$.} refuting this conjecture for $n=5$ and $p=0.4$. Similar to~\citep{gopi2013stability}, they also provide a partial positive result showing that for all odd $n$, the conjecture holds in a small neighborhood of $p$ around $0$.

\subsection{Our Contributions}

In this work, we provide a unified understanding of both these seemingly unrelated conjectures, and nearly-tight characterization of the noise parameter regimes under which the respective conjectures are tight. Our main contributions are summarized as follows:

\paragraph{When the conjectures are false?}
We show, in \Cref{thm:main-counterexample}, that the counterexample for the \mlsc conjecture by \citet{gopi2013stability} also refutes the \nicdc conjecture for all $n \ge 5$. Furthermore, we show that this example refutes the \mlsc conjecture for all $\rho \in (0, 1 - \frac{8 + o(1)}{n^2})$ and refutes the \nicdc conjecture for all $p \in (\frac{4 + o(1)}{n^2}, \frac12)$.

\paragraph{When the conjectures are true?}
We complement the counterexamples by obtaining nearly matching positive bounds in \Cref{thm:main-example}: the \mlsc conjecture  holds for $\rho \in (1 - \frac{8 - o(1)}{n^2}, 1)$ and the \nicdc conjecture holds for $p \in (0, \frac{2 - o(1)}{n^2})$. Our result for the \mlsc conjecture is in fact stronger in that we show the \mlsc conjecture holds
for all unate\footnote{A Boolean function $f : \{-1, 1\}^n \to \{-1, 1\}$ is \emph{unate} if it is monotone, either increasing or decreasing, in each $x_i$.}
functions, of which LTFs are a special case. Additionally, the \nicdc conjecture when restricted to unate functions holds for all $p \in (0, \frac{4 - o(1)}{n^2})$.
Using our techniques, we also show that the conjectures hold for $n=3$ for the entire regime of parameters (\Cref{thm:n3}).\footnote{
\citet{BiswasS23} already showed
this for the \mlsc conjecture; see the proof of~\Cref{thm:n3}.
}

\subsection{Related Work \& Discussion}\label{sec:related}

For any distribution $P$ over $\cX\times \cY$, the maximal correlation coefficient $\rho(P)$, studied by~\citet{hirschfeld1935connection,gebelein1941statistische,Renyi1959measures} is the maximum value of $\Ex[f(x) g(y)]$ subject to $\Ex[f(x)] = \Ex[g(y)] = 0$ and $\Ex[f(x)^2] = \Ex[g(\by)^2] = 1$ over $f : \cX\to \R$ and $g : \cY \to \R$, where the expectation is over $(x, y) \sim P$. The difference between $\rho^\star(P)$ and $\rho(P)$ is that the range of $f$ and $g$ is $[-1, 1]$ in the former, but $\R$ in the latter.\footnote{Another difference is that in $\rho(P)$, the domain of $f$ is a single $x$ (similarly for $g$), however, $\rho(P)$ exhibits a ``tensorization property'', namely, $\rho(P) = \rho(P^{\otimes n})$.}
Thus, it follows that $\rho^\star(P) \le \rho(P)$. \citet{witsenhausen75sequences} provided a construction of $f$ and $g$ to show that
\begin{align}
\textstyle
\frac{2}{\pi} \arcsin(\rho(P)) ~\le~ \rho^\star(P) ~\le~ \rho(P).\label{eq:witsenhausen}
\end{align}
It is challenging to determine $\rho^*(P)$, and is not known exactly even for some simple distributions $P$.
Only recently, it was shown that $\rho^*(P)$ is even  computable~\citep{ghazi16decidability}; subsequently, this was extended to correlations over larger alphabet size~\citep{de17noise,de18noninteractive,ghazi18dimension}.

\paragraph{\boldmath $\DSBS_\rho$.} The maximal correlation $\rho(\DSBS_\rho) = \rho$ and thus, we have $\rho^*(\DSBS_\rho) \le \rho$. Moreover the dictator functions achieve equality thereby establishing their optimality.

Witsenhausen's construction underlying \eqref{eq:witsenhausen} is identical to $f = g = \Maj_n$, in the limit as $n \to \infty$, it holds that $\Stab_\rho(\Maj_n)$ approaches $\frac 2\pi \arcsin \rho$ from above~\citep{odonnell2003thesis}, and hence $\Stab_\rho(\Maj_n) \ge \frac 2\pi \arcsin \rho$ for all $n$.

Somewhat unrelated,~\citet{MOO10} introduced the {\em invariance principle} and showed that for functions $f$ with ``low influence'', it holds that $\Stab_\rho(f) \le \Stab_{\rho}(\Maj_n) + o_n(1)$, i.e., the \mlsc conjecture holds for   low-influence functions.

Given all the results thus far, a reasonable refined version of the \mlsc conjecture that we believe captures the spirit of the original formulation is as follows. This formulation was also stated by \cite{gopi2013stability}.

\begin{conjecture}[Refined \mlsc~\citep{gopi2013stability}]\label{conj:mils-refined}
For all $\rho \in (0, 1)$, odd $n$, and unbiased LTF $f : \{-1, 1\}^n \to \{-1, 1\}$, it holds that $\Stab_\rho(f) \ge \frac{2}{\pi} \arcsin(\rho)$.
\end{conjecture}

\noindent {\em Remark.} It is tempting to conjecture that the \mlsc conjecture holds for all unate functions, not just LTFs, given that our positive result shows this for $\rho \in (1 - O(1/n^2), 1)$. However, this turns out to be false. \cite{benjamini1999noise} showed that the noise stability of the so-called  Tribes function (a balanced monotone function) approaches $0$ as $n \to \infty$ for $\rho = 1 - \Omega(1/\log n)$.

\paragraph{\boldmath $\SBESp$.}
It is easy to see that the maximal correlation coefficient $\rho(\SBESp) = \sqrt{p}$. In this case, Witsenhausen's construction corresponds to $f = \Maj_n$, and so we get $\Phi_p(f)$ approaches $\frac 2\pi \arcsin(\sqrt{p})$ as $n \to \infty$. Furthermore, $\Phi_p(\Maj_n)$ approaches this limit from below when $p < 1/2$ (and approaches from above when $p > 1/2$).\footnote{This is easy to show, and is similar in spirit to Condorcet's Jury Theorem. We omit the proof.}

Somewhat unrelated, \citet{mossel2010gaussian} extended the invariance principle of \citet{MOO10} and showed, as an application, that for any low-influence  function $f$, it holds that $\Phi_p(f) \le \Phi_p(\Maj_n) + o_n(1)$.

Given all the results thus far, a reasonable refined version of the \nicdc conjecture that we believe captures the spirit of the original formulation is as follows---it is basically saying that for $p < 1/2$, Witsenhausen's construction is optimal for NICD for $\SBES_p$.

\begin{conjecture}[Refined \nicdc]\label{conj:nicd-erasures-refined}
For all $p \in (0, 1/2)$, odd $n$, and unbiased function $f : \{-1, 1\}^n \to \{-1, 1\}$, it holds that $\Phi_p(f) \le \frac{2}{\pi} \arcsin(\sqrt{p})$.
\end{conjecture}

\noindent {\em Remark.} Note however, that the formulation in terms of $\Phi_p$ is not equivalent to the NICD setting because $g(\by) := \sign(\Ex[f(\bx) | \by])$ need not satisfy $\Ex[g(\by)] = 0$. So a weaker variant, implied by the above conjecture, is that $\Ex[f(\bx) g(\by)] \le \frac{2}{\pi} \arcsin(\sqrt{p})$ whenever $\Ex [f(\bx)] = \Ex [g(\by)] = 0$.

\section{Correlation Distillation Conjectures}\label{sec:background}

Let $[n]$ denote $\{1, \ldots, n\}$.  For $\bx = (x_1, \ldots, x_n)$ and for $S \subseteq [n]$, let $\bx_{S}$ denote the subsequence of $\bx$ indexed by $S$ and let $\bone_S$ denote the all-ones vector with coordinates indexed by $S$.

\subsection{Majority is Least Stable (\mlsc)}\label{subsec:mils}

\begin{conjecture}[\mlsc,~\citet{benjamini1999noise}]
For any odd integer $n$, any LTF $f: \{-1, 1\}^n \to \{-1, 1\}$ and $\rho \in [0, 1]$, it holds that $\Stab_\rho[f] \geq \Stab_\rho[\Maj_n]$.
\end{conjecture}

\citet{gopi2013stability} proposed a counterexample to the conjecture for all odd $n \ge 5$. To define this counterexample, let $h = (n - 1)/2$, and partition $[n]$ into $A = [h + 1] = \{ 1, \ldots, h+1 \}$ and $B = [n] \setminus A = \{ h+2, \ldots, n \}$. Let $\bhole = \bone_A \circ -\bone_B$. Define $\counterf: \{-1, 1\}^n \to \{-1, 1\}$ as
\begin{align} \label{def:counter-exp}
\counterf(x) = \begin{cases}
- \Maj_n(x) & x \in \{ \bhole, -\bhole \}, \\
\Maj_n(x) & \text{ otherwise.}
\end{cases}
\end{align}
In other words, $\counterf$ is almost the same as $\Maj_n$ except at two antipodal points at the boundary that are negated. As noted in \citep{gopi2013stability}, $g_n$ is indeed an unbiased LTF, as it can be expressed as $\sign\Paren{a \cdot \sum_{i \in A} x_i + b \cdot \sum_{i \in B} x_i}$ for any odd positive integers $a, b$ such that $\frac{b}{a} \in \Paren{1 + \frac{1}{h}, 1 + \frac{1}{h - 2}}$. (For example, $a = 2h - 1 = n - 2$ and $b = 2h - 3 = n - 4$.) It was shown that this is indeed a counterexample to the \mlsc
conjecture, as stated more formally below.

\begin{theorem}[\cite{gopi2013stability}] \label{thm:gopi-negative}
For every odd $n \geq 5$, there is $\rho_n > 0$ such that $\Stab_\rho[\counterf] < \Stab_\rho[\Maj_n]$ for all $\rho \in (0, \rho_n)$.
\end{theorem}

Later,~\citet{BiswasS23} show that the \mlsc conjecture holds for all $\rho \in (0, 1)$ when $n \leq 3$.
Thus, the above counterexample is minimal.
Additionally, they also independently discovered a different counterexample and proved a statement similar to \Cref{thm:gopi-negative}. Nevertheless, an interesting question remains: Could the counterexample be extended to hold for all $\rho \in (0, 1)$?  \citet{gopi2013stability} partially answered this question, by showing that, for any fixed $n$, there exists a small neighborhood of $\rho$ around 1 in which the \mlsc
conjecture holds:

\begin{theorem}[\cite{gopi2013stability}] \label{thm:gopi-positive}
For every odd $n$, there exists $\zeta_n > 0$ such that, for all $\rho \in [1 - \zeta_n, 1]$, it holds that $\Stab_\rho[f] \geq \Stab_\rho[\Maj_n]$ for all LTF $f: \{-1, 1\}^n \to \{-1, 1\}$.
\end{theorem}
\noindent However, this does not answer, e.g., if for a fixed $\rho = 0.99$, does there exist $n$ for which the 
\mlsc conjecture is false?

\subsection{NICD for Erasures (\nicdc)}

\begin{conjecture}[\nicdc,~\citet{odonnell12new}]
For any odd integer $n$, any unbiased $f$ and $p \in (0, 1/2)$, it holds that $\Phi_p[f] \leq \Phi_p[\Maj_n]$.
\end{conjecture}

Recently,~\citet{ivanisvili2025nicd} gave a (numerical) counterexample to this conjecture,
for a specific value of $n, p$:

\begin{theorem}[\cite{ivanisvili2025nicd}] \label{thm:ix-negative}
For $n = 5$ and $p = 0.4$, there is an unbiased $f: \{-1, 1\}^n \to \{-1, 1\}$ such that $\Phi_p[f] > \Phi_p[\Maj_n]$.
\end{theorem}
\noindent And similar to \Cref{thm:gopi-positive}, they showed that the \nicdc conjecture holds for $p$ close to 0:

\begin{theorem}[\cite{ivanisvili2025nicd}] \label{thm:ix-positive}
For any odd integer $n$, there exists $\gamma_n > 0$ such that, for all unbiased $f$ and $p \in [0, \gamma_n]$, it holds that $\Phi_p[f] \leq \Phi_p[\Maj_n]$.
\end{theorem}

\section{Our Results}
Our first result is to unify these two lines of work and show that Gopi's counterexample violates the \nicdc conjecture as well for all $n \geq 5$, and furthermore provides a strengthening of \Cref{thm:gopi-negative} by showing that the counterexample refutes
the \mlsc conjecture for almost all $\rho$ (i.e., $\rho < 1 - \Omega(1/n^2)$):
\begin{theorem} \label{thm:main-counterexample}
Let $\eps_n := \frac{4}{(n-3)^2 + 6} = \frac{4 + o(1)}{n^2}$.
There is a family $\{ g_n \}$ of unbiased LTFs such that the following holds for all odd integers $n \geq 5$:
\begin{enumerate}
\item[(i)] $\Stab_\rho[\counterf] < \Stab_\rho[\Maj_n]$ for all $\rho \in \Paren{ 0, 1 - 2\eps_n}$. 
\item[(ii)] $\Phi_p[\counterf] > \Phi_p[\Maj_n]$ for all $p \in \Paren{ \eps_n, \frac{1}{2} }$.
\end{enumerate}
\end{theorem}

We complement this result and show almost matching bounds: (i) the \mlsc conjecture holds  all unate functions (not just LTFs) for $\rho$ close to 1 and (ii) the 
\nicdc conjecture holds for $p$ close to $0$, providing quantitative improvements over \Cref{thm:gopi-positive,thm:ix-positive} where an explicit bound on the regime of parameters is not provided.

\begin{theorem}
\label{thm:main-example}
Let $\gamma_n := \frac{4}{(n + 7)(n - 1)} = \frac{4 - o(1)}{n^2}$.  
For every odd $n \ge 5$ and any unate
unbiased $f$, we have the following:
\begin{enumerate}
\item[(i)] $\Stab_\rho[\Maj_n] \le \Stab_\rho[f]$ for all $\rho \in \left(1 - 2\gamma_n, 1\right)$. 
\item[(ii)] $\Phi_p[\Maj_n] \ge \Phi_p[f]$ for all $p \in \Paren{0, \gamma_n}$.
\end{enumerate}
Moreover, for every odd $n \ge 5$ and any (possibly non-unate) unbiased $f$, we have:
\begin{enumerate}
\item[(iii)] $\Phi_p[\Maj_n] \ge \Phi_p[f]$ for all $p \in \Paren{0, \gamma'_n}$ for $\gamma'_n := \frac{2}{(n+2)(n-1)} \geq \frac{2}{n^2}$.
\end{enumerate}
In all cases, the inequality is in fact strict as long as $f$ is not $\Maj_n$ (up to sign flips).\end{theorem}

\noindent Finally, using similar techniques, we show that both conjectures hold unconditionally for $n=3$.
\begin{theorem}
\label{thm:n3}
The following statements hold:
\begin{enumerate}
\item $\Stab_\rho[\Maj_3] \le \Stab_\rho[f]$ holds for all unbiased LTF $f$ and $\rho \in (0, 1)$, and
\item $\Phi_p[\Maj_3] \ge \Phi_p[f]$ holds for all unbiased $f$ and $p \in (0, 1/2)$.
\end{enumerate}
In both cases, the inequality is in fact strict as long as $f$ is not $\Maj_3$ (up to sign flips).
\end{theorem}

\noindent While the case of \mlsc for $n=3$ was also shown by \cite{BiswasS23}, their proof involved an exhaustive enumeration over functions on $3$ variables, and evaluating which ones are LTFs, whereas our proof is free of case analysis, and uses techniques developed to show \Cref{thm:main-example}.

\section{When the conjectures are false? (Proof of \texorpdfstring{\Cref{thm:main-counterexample}}{Theorem~\ref{thm:main-counterexample}})}

\subsection{Useful Properties of Binomials}

We start by a couple of useful lemmata on binomial distributions. First is a lemma that relates $\Bin(h+1, p)$ and $\Bin(h, p)$:
\begin{lemma} \label{prop:binom-comp}
For any $p \in [0, 1]$ and $h \in \N$, we have
\begin{align*}
\Ex_{\substack{u \sim \Bin(h + 1, p) \\ u' \sim \Bin(h, p)}}\Brack{(-1)^{\bone[u > u']}} = (1 - 2p) \Pr_{v, v' \sim \Bin(h, p)}[v = v'] = (1 - 2p) \Pr_{w, w' \sim \Bin(h, 1 - p)}[w = w']
\end{align*}
\end{lemma}

\begin{proof}
Write $u = v + m$ where $v \sim \Bin(h, p)$ and $m \sim \Ber(p)$ are independent. Since $u', v'$ are identically distributed, we have
\begin{align*}
\Ex_{\substack{u \sim \Bin(h + 1, p) \\ u' \sim \Bin(h, p)}}\Brack{(-1)^{\bone[u > u']}}
&= \Ex_{\substack{v, u' \sim \Bin(h, p) \\ m \sim \Ber(p)}}\Brack{(-1)^{\bone[v + m > u']}} \\
&= \Ex_{\substack{v, u' \sim \Bin(h, p) \\ m \sim \Ber(p)}}\Brack{\frac{1}{2}\Paren{(-1)^{\bone[v + m > u']} + (-1)^{\bone[u' + m > v]}}}.
\end{align*}
Notice that the term $(-1)^{\bone[v + m > u']} + (-1)^{\bone[u' + m > v]} = 0$ unless $u' = v$. Thus, the RHS equals
\begin{align*}
\Pr_{u', v'}[u' = v'] \cdot \Ex_{m \sim \Ber(p)}\Brack{(-1)^{\bone[m > 0]}} = \Pr_{u',v'}[u' = v'] \cdot (1 - 2p),
\end{align*}
which yields the first equality.

The second equality follows simply by substituting $w = h - v, w' = h - v'$.
\end{proof}

The second lemma gives a lower bound on the probability that two i.i.d. binomial random variables are equal. The form of the lower bound might look peculiar, but this is exactly what we will need subsequently.

\begin{lemma} \label{prop:binom-eq-lb}
For any $h \in \N$ such that $h \geq 2$ and any $q \in \Paren{\frac{1}{(h - 1)^2 + 2}, \frac{1}{2}}$, we have 
\begin{align*}
(1 - 2q)\Pr_{w, w' \sim \Bin(h, q)}[w = w'] + q^{2h + 1} - (1 - q)^{2h + 1} > 0.
\end{align*}
\end{lemma}

\begin{proof}
We have
\begin{align*}
\Pr_{w, w' \sim \Bin(h, q)}[w = w'] = \sum_{i=0}^{h} \Pr_w[w = i] \cdot \Pr_{w'}[w' = i]
&= \sum_{i=0}^h \Paren{\binom{h}{i} q^i (1 - q)^{h-i}}^2 \\
&\geq (1 - q)^{2h} + h^2 q^2(1 - q)^{2h-2} + q^{4}(1-q)^{2h-4},
\end{align*}
where the inequality follows by considering only 
the terms corresponding to $i \in \{0, 1, 2 \}$, and using the rough lower bound $\binom{h}{2} \geq 1$.

Using $q \leq 1/2$ and the above inequality, we can thus bound the desired quantity as follows:
\begin{align*}
&(1 - 2q)\Pr_{w, w' \sim \Bin(h, q)}[w = w'] + q^{2h + 1} - (1 - q)^{2h + 1} \\ &\geq (1 - 2q)\left((1 - q)^{2h} + h^2 q^2(1 - q)^{2h-2} + q^{4}(1 - q)^{2h-4} - \sum_{j=0}^{2h} q^j (1 - q)^{2h - j} \right) \\
&\geq (1 - 2q)\Paren{(h^2 - 1)q^2(1 - q)^{2h - 2} - q(1 - q)^{2h - 1} - (2h - 3)q^3(1 - q)^{2h - 3}} \\
&\geq (1 - 2q)\Paren{(h^2 - 2h + 2)q^2(1 - q)^{2h - 2} - q(1 - q)^{2h - 1}} \\
&= (1 - 2q)q(1 - q)^{2h - 2} \Paren{(h^2 - 2h + 3)q - 1} \\
& > 0,
\end{align*}
where the last inequality is due to our assumption that $q \in \Paren{\frac{1}{(h - 1)^2 + 2}, \frac{1}{2}}$.
\end{proof}

\subsection{Characterization of the Gaps}

We now establish a simple formula of the differences $\Phi_q[\counterf] - \Phi_q[\Maj_n]$ and $\Stab_{1 - 2q}[\Maj_n] - \Stab_{1 - 2q}[\counterf]$ based on binomial distributions. In fact, the two quantities are equal, up to an appropriate scaling, as stated below.

\begin{theorem} \label{thm:counterexample-formula}
For all odd $n \geq 3$ and $q \in (0, 1)$, it holds that
\begin{align*}
2^{n-2} \cdot (\Phi_q[\counterf] - \Phi_q[\Maj_n])
&= 2^{n-3} \cdot (\Stab_{1 - 2q}[\Maj_n] - \Stab_{1 - 2q}[\counterf]) \\
& = (1 - 2q)\Pr_{z, z' \sim \Bin\Paren{\frac{n - 1}{2}, q}}\left[z = z'\right] + q^n - (1 - q)^n.
\end{align*}
\end{theorem}

\noindent Note that \Cref{thm:main-counterexample} immediately follows from \Cref{thm:counterexample-formula} and \Cref{prop:binom-eq-lb} (using $h = \frac{n-1}{2}$).

\begin{proof}[Proof of \Cref{thm:counterexample-formula}]
We will separately derive the formula for $\Stab$ and $\Phi$. Let $\delta_n : \{-1, 1\}^n \to \{-2, 0, 2\}$ be given as $\delta_n(\bx) = \Maj_n(\bx) - \counterf(\bx)$. From definition of $\counterf$ \eqref{def:counter-exp}, it follows that $\delta_n(e) = 2$, $\delta_n(-e) = -2$ and $\delta_n(\bx) = 0$ for all other $\bx$.

\paragraph{\boldmath Derivation for $\Stab_\rho$.}
Let $\rho = 1 - 2q$. For $(\bx, \by) \sim \DSBS_\rho$, we can write the difference between stabilities as
\begin{align*}
\Stab_\rho[\Maj_n] -  \Stab_\rho[\counterf] &= \Ex\Brack{ \Maj_n(\bx)\Maj_n(\by) - \counterf(\bx)\counterf(\by) } \\
&= \Ex\Brack{ \Maj_n(\bx)\Maj_n(\by) - \Paren{ \Maj_n(\bx) - \difff(\bx) } \Paren{ \Maj_n(\by) - \difff(\by) } } \\
&= 2 \Ex[\difff(\bx) \Maj_n(\by)] - \Ex[\difff(\bx) \difff(\by)].
\end{align*}
Since $\difff$ is only supported at $\bhole$ and $-\bhole$, it is easy to see that
\begin{align*}
\Ex[\difff(\bx) \difff(\by)] = \frac{1}{2^{n-3}} \Paren{ \Paren{ \frac{1 + \rho}{2} }^n - \Paren{ \frac{1 - \rho}{2} }^n } = \frac{1}{2^{n-3}}\Paren{(1 - q)^n - q^n}.
\end{align*}
To compute the first expectation, using the fact that both $\difff$ and $\Maj_n$ are odd functions, we have
\begin{align*}
\Ex[\difff(\bx) \Maj_n(\by)] &= \Ex[\difff(\bx) \Maj_n(\by) \mid x_1 = 1] = \frac{1}{2^{n-2}} \Ex[\Maj_n(\by) \mid \bx = \bhole]\,,
\end{align*}
where the second inequality is due to the fact that, conditioned on $x_1 = 1$, $\difff(\bx)$ is only non-zero when $\bx = \bhole$. Now, let $n_A = |\{i \in A \mid y_i = 1\}|$ and  $n_B = |\{i \in B \mid y_i = -1\}|$. 
We have that $\Maj_n(\by) = (-1)^{\bone[n_A > n_B]}$. Combining all these, we have
\begin{align*}
\Stab_\rho[\Maj_n] -  \Stab_\rho[\counterf] &= \frac{1}{2^{n-3}} \Paren{ \Ex\Brack{(-1)^{\bone[n_A > n_B]}} + q^n - (1 - q)^n. }
\end{align*}
Notice that, conditioned on $\bx = e$, we have $n_A \sim \Bin\Paren{h + 1, \frac{1 + \rho}{2}} = \Bin\Paren{h+1, 1 - q}$ and $n_B \sim \Bin\Paren{h, \frac{1 + \rho}{2}} = \Bin\Paren{h, 1 - q}$. Thus, from \Cref{prop:binom-comp}, we have
\begin{align*}
\Stab_\rho[\Maj_n] -  \Stab_\rho[\counterf] = \frac{1}{2^{n-3}}\Paren{(1 - 2q)\Pr_{w, w' \sim \Bin(h, q)}[w = w'] + q^n - (1 - q)^n},
\end{align*}
as desired.

\paragraph{\boldmath Derivation for $\Phi_p$.}
We are interested in $\Phi_p[\counterf] - \Phi_p[\Maj_n]$ for $p = q$. Let $S \subseteq [n]$ be generated by including each element of $[n]$ independently with probability $p$, and let $\bx_S \in \{-1, 1\}^S$ be uniformly at random. We also write $\oS$ to denote $[n] \setminus S$. Note that $\Phi_p$ can be written as
\begin{align*}
\Phi_{p}[f] 
:= 
\Ex_{S} \left[ \Ex_{\bx_S} \left[ \left | \Ex_{\bx_{\bar{S}}} [f(\bx)] \right | \right] \right]
\end{align*}
Thus, we can write the desired quantity as
\begin{align*}
\Phi_p[\counterf] - \Phi_p[\Maj_n] &= \Ex_S \Ex_{\bx_S} \Brack{\frac{1}{2^{n - |S|}} \cdot h_S(\bx_S)},
\end{align*}
where
\begin{align*}\textstyle
h_S(\bx_S) := \Abs{\sum_{\bx_{\oS} \in \{-1, 1\}^{\oS}} \ \counterf(\bx_S \circ \bx_{\oS})} - \Abs{\sum_{\bx_{\oS} \in \{-1, 1\}^{\oS}} \ \Maj_n(\bx_S \circ \bx_{\oS})}.
\end{align*}
For notational convenience, write $\bhole_S$ as the restriction of $\bhole$ on $S$.
Now, if $\bx_S \notin \{\bhole_S, -\bhole_S\}$, then we immediately have $\counterf(\bx_S \circ \bx_{\oS}) = \Maj_n(\bx_S \circ \bx_{\oS})$, and $h_S(\bx_S) = 0$. Furthermore, notice that if $S = \emptyset$, we always have $h_S(\bx_S) = 0$. Thus, the above term can be written as
\begin{align}
\Phi_p[\counterf] - \Phi_p[\Maj_n]  &= \Ex_S \Brack{\frac{1}{2^{n}} \cdot \Paren{h_S\Paren{\bhole_S} + h_S\Paren{-\bhole_S}}} = \frac{1}{2^{n-1}} \Ex_S\Brack{h_S(\bhole_S)}, \label{eq:diff-simplified}
\end{align}
where in the second equality, we use the fact that $\counterf$ and $\Maj_n$ are odd and thus $h_S$ is even.

Finally, to determine the value of $h_S(\bhole_S)$, observe that, unless $S = \emptyset$, we have $$\sum_{\bx_{\oS} \in \{-1, 1\}^{\oS}} \counterf(\bhole_S \circ \bx_{\oS}) = -2 + \sum_{\bx_{\oS} \in \{-1, 1\}^{\oS}} \Maj_n(\bhole_S \circ \bx_{\oS}).$$ This means that $h_S(\bhole_S)$ is completely determined by the sign of $\sum_{\bx_{\oS} \in \{-1, 1\}^{\oS}} \Maj_n(\bx_S \circ \bx_{\oS})$:
\begin{align*}
h_S(\bhole_S) &= 
\begin{cases}
0 &\text{ if } S = \emptyset \text{ or } S = [n],\\
-2 &\text{ if } \sum_{\bx_{\oS} \in \{-1, 1\}^{\oS}} \Maj_n(\bhole_S \circ \bx_{\oS}) > 0,\\
2 &\text{ otherwise.} 
\end{cases}
\end{align*}
Let $n_A := |S \cap A|, n_B := |S \cap B|$.
Finally, notice that $\Maj_n(\bhole_S \circ \bx_{\oS})$ is simply $\sign\Paren{n_A - n_B + \sum_{i \in \oS} x_i}$.
Thus, we have
\begin{align*}
h_S(\bhole_S) &= 
\begin{cases}
0 &\text{ if } S = \emptyset \text{ or } S = [n],\\
2 \cdot (-1)^{\bone[n_A > n_B]} &\text{ otherwise.} 
\end{cases}
\end{align*}
Plugging this back into \eqref{eq:diff-simplified}, we have
\begin{align*}
\Phi_p[\counterf] - \Phi_p[\Maj_n] = \frac{1}{2^{n - 2}} \Paren{\Ex\Brack{(-1)^{\bone[n_A > n_B]}} + p^n - (1 - p)^n}.
\end{align*}
Notice that $n_A \sim \Bin(h+1, p), n_B \sim \Bin(h, p)$ are independent. Thus, applying \Cref{prop:binom-comp}, we get
\begin{align*}
\Phi_p[\counterf] - \Phi_p[\Maj_n] = \frac{1}{2^{n - 2}} \Paren{(1 - 2p) \Pr_{w, w' \sim \Bin(h, p)}[w = w'] + p^n - (1 - p)^n}. & &\qedhere
\end{align*}
\end{proof}

\section{When the conjectures are true? (Proof of \texorpdfstring{\Cref{thm:main-example}}{Theorem~\ref{thm:main-example}})}
\label{sec:tight}

For any function $f: \{-1, 1\}^n \to \{-1, 1\}$, let $\diffset_f := \{\bx \in \{-1, 1\}^n \mid f(\bx) \ne \Maj_n(\bx)\}$ and $\mu_f := \Pr_{\bx \sim \{-1, 1\}^n}[f(\bx) \ne \Maj_n(\bx)] = \frac{|E_f|}{2^n}$.

\subsection{Useful Lemmata on Gaps}

We first obtain a lower bound on the difference
between the first-level Fourier coefficients of 
Majority and $f$ in terms of their disagreement.
We state the lemma below without defining Fourier coefficients for simplicity.
\begin{lemma} 
\label{lem:maj-diff}
For any function $f: \{-1, 1\}^n \to \{-1, 1\}$, it holds that
$$\sum_{i \in [n]} \ \Ex_{\bx \sim \{-1, 1\}^n}\Brack{x_i \cdot (\Maj_n(\bx) - f(\bx))} \geq 2 \cdot \mu_f.$$
\end{lemma}

\begin{proof}
We have
\begin{align*}
\sum_{i \in [n]} \Ex_{\bx \sim \{-1, 1\}^n}\Brack{x_i \cdot (\Maj_n(\bx) -  f(\bx))} 
&= \Ex_{\bx \sim \{-1, 1\}^n}\Brack{\Paren{\sum_{i \in [n]} x_i} \cdot (\Maj_n(\bx) - f(\bx))} \\
&\geq 2 \cdot \Pr_{\bx \sim \{-1, 1\}^n}\Brack{\bx \in \diffset_f} = 2 \cdot \mu_f,
\end{align*}
where in the inequality we use the fact that $\Paren{\sum_{i \in [n]} x_i}$ and $\Maj_n(x)$ have the same sign, 
(and that the former is a non-zero integer), 
which means that $\Paren{\sum_{i \in [n]} x_i} \cdot (\Maj_n(x) - f(x)) \geq 2$ if $f(x) \ne \Maj_n(x)$.
\end{proof}
We next prove a simple algebraic fact that we will use later.
\begin{lemma}
\label{lem:qvalue}
For $n \geq 5$, 
\[
q < \frac{4}{(n+7)(n-1)}
\implies
\left( 2(1-q)^{n-1} - q(1-q)^{n-2}\frac{(n-1)^2}{2} - q^2 \cdot 2\binom{n}{3} \right) > 0.
\]
\end{lemma}
\begin{proof}
From $q < \frac{4}{(n+7)(n-1)}$, we have
\begin{align*}
& 2(1-q)^{n-1} - q(1-q)^{n-2}\frac{(n-1)^2}{2} - q^2 \cdot 2\binom{n}{3} \\
&= (1 - q)^{n - 2}\Paren{2(1 - q) - q\frac{(n - 1)^2}{2}} - q^2 \cdot 2\binom{n}{3} \\
&= (1 - q)^{n - 2}\Paren{2 - q\frac{(n + 1)^2}{2}} - q^2 \cdot 2\binom{n}{3} \\
&\stackrel{(*)}{\geq} (1 - (n-2)q)\Paren{2 - q\frac{(n + 1)^2}{2}} - q^2 \cdot 2\binom{n}{3} \\
&= 2 - q\frac{(n+1)^2}{2} - 2q(n-2) + q^2(n-2)\frac{(n+1)^2}{2} - q^2 \cdot 2\binom{n}{3} \\
&= 2 - q \cdot \frac{(n+7)(n-1)}{2} + q^2 \cdot (n - 2) \cdot \left(\frac{(n + 1)^2}{2} - \frac{n(n-1)}{3} \right) \\
& \stackrel{(**)}{>} 2 - q \cdot \frac{(n+7)(n-1)}{2}  > 0,
\end{align*}
where in (*), we applied Bernoulli's inequality, $(1-q)^{n-2} \ge 1 - (n-2)q$, to the first term and 
in (**), we used that the numerator
of the dropped term is positive
for $n \geq 5$.
\end{proof}

For any $f : \{-1, 1\}^n \to \{-1, 1\}$ we will consider the ``gap'' between $f$ and $\Maj_n$ that has a very specific form. In particular, the gap will have the following binomial-expansion-like form:
\begin{align} \label{eq:gap-generic}
G_q[f] = \sum_{k=1}^{n} \Brack{\textstyle q^k (1-q)^{n-k} \cdot \sum_{S \in \binom{[n]}{k}} G_S(f)},
\end{align}
for some $G_S(f) \in \R$.
We prove a generic bound on $G_q[f]$ in terms of the subset expressions $G_S(\cdot)$.  
\begin{lemma}
\label{lem:gap}
Suppose there is a constant $c > 0$ such that:
\begin{enumerate}[leftmargin=7mm]
\item[(i)] For $i \in [n]$, $G_{\{i\}}(f) = c \cdot \Ex_{\bx}[x_i \cdot (\Maj_n(\bx) - f(\bx))]$.
\item[(ii)] For $i \ne j \in [n]$, 
$G_{\{i,j\}}(f) \ge 
c \cdot \Paren{\frac{1}{4} \Ex_{\bx}[(x_i+x_j)(\Maj_n(\bx)-f(\bx)) \mid x_i=x_j] - 
\Pr[\bx \in E_f \mid x_i \ne x_j]}.$
\item[(iii)] For $S$ such that $|S| \geq 3$, $G_S(f) \ge -2c \cdot \Pr[\bx \in E_f]$.
\end{enumerate}
Then, $G_q[f] > 0$ for $q < \frac{4}{(n+7)(n-1)}$.
\end{lemma}
\begin{proof}
We lower bound the LHS using our assumptions on each term on the RHS.  For the $k = 1$ term, we sum over all singletons and apply~\Cref{lem:maj-diff} to obtain
\[
\sum_{i \in [n]} G_{\{i\}}(f) 
= c \cdot \sum_{i \in [n]} \Ex_{\bx}[x_i \cdot (\Maj_n(\bx) - f(\bx))] 
\ge 2c \cdot \mu_f.
\]
Thus, the contribution to $G_q[f]$ from the $k = 1$ term is at least $2c \cdot \mu_f \cdot q(1-q)^{n-1}$.

For the term corresponding to $k = 2$, we sum over all $\binom{n}{2}$ pairs.  Note that the bound in our assumption has two components.  Considering
the first component, we obtain
\begin{align*}
c \cdot 
\sum_{i < j} \frac{1}{4} \Ex_{\bx}[(x_i+x_j)(\Maj_n(\bx)-f(\bx)) \mid x_i = x_j] 
& = \frac{c}{2} \cdot 
\sum_{i < j} 
\Ex_{\bx}[(x_i+x_j)(\Maj_n(\bx)-f(\bx))] \\
& = c \cdot \frac{n-1}{2} 
\sum_{i \in [n]} \Ex_{\bx} [x_i (\Maj_n(\bx) - f(\bx))] \\
& \ge c(n-1)\mu_f,
\end{align*}
where the first equality follows by observing that if $x_i \ne x_j$ then $x_i + x_j = 0$, so we can drop the conditioning on $x_i = x_j$ by gaining a factor of $2$, and the last step follows from \Cref{lem:maj-diff}.  
Now we consider the second component.  
\begin{align*}
-c \cdot \sum_{i < j} \Pr \Brack{\bx \in \diffset_f  ~\middle|~ x_i \ne x_j} 
&= -2c \cdot \sum_{i < j} \Pr \Brack{\bx \in \diffset_f  \land x_i \ne x_j} \\
&= -2c \binom{n}{2} \cdot \Pr_{\bx, S = \{i, j\}} \Brack{\bx \in \diffset_f  \land x_i \ne x_j} \\
&= -2c \binom{n}{2} \cdot \Pr_{\bx, S = \{i, j\}} \Brack{x_i \ne x_j ~\middle|~ \bx \in \diffset_f} \cdot \mu_f \\
&\geq -\frac{c}{2} (n^2 - 1) \cdot \mu_f,
\end{align*}
where the last inequality follows from the fact that, 
for any $\bx \in \{-1, 1\}^n$, the probability (over random $\{i, j\} \in \binom{[n]}{2}$) that $x_i \ne x_j$ is at most $\frac{\Paren{\frac{n+1}{2}}\Paren{\frac{n-1}{2}}}{\binom{n}{2}} = \frac{(n^2 - 1)}{4 \binom{n}{2}}$.

Thus, the contribution to $G_q[f]$ from the $k=2$ term is at least 
\[
q^2(1-q)^{n-2} \left( c(n-1)\mu_f - c\frac{n^2-1}{2}\mu_f \right) = -c \cdot q^2(1-q)^{n-2}\frac{(n-1)^2}{2}\mu_f.
\]

For the terms corresponding to $k \ge 3$, we lower bound the contribution to $G_q[f]$ as follows: 
\begin{align*}
\sum_{k=3}^n q^k (1-q)^{n-k} \sum_{S \in \binom{[n]}{k}} G_S(f) & \ge -2c \sum_{k=3}^n q^k (1 - q)^{n-k} \sum_{S \in \binom{[n]}{k}} \Pr[x \in E_f] \\ 
&\ge -2c \cdot \mu_f \sum_{k=3}^n q^k (1 - q)^{n-k} \binom{n}{k} \\
& \ge -2c \cdot \mu_f q^3 \binom{n}{3},
\end{align*}
where the last inequality is due to the union bound: $\sum_{k=3}^n q^k (1 - q)^{n-k} \binom{n}{k}$ is the probability that there are at least 3 successes from $n$ i.i.d. Bernoulli random variables with success probability $q$.

Summing the above contributions and applying
\Cref{lem:qvalue} with the choice of $q$, we obtain
\begin{align*}
G_q[f] & \geq c \cdot \mu_f \cdot q \left( 2(1-q)^{n-1} - q(1-q)^{n-2}\frac{(n-1)^2}{2} - q^2 \cdot 2\binom{n}{3} \right) > 0.
\qedhere
\end{align*}
\end{proof}

\subsection{Bounds on the Gaps}

\paragraph{\boldmath Derivation for $\Stab_\rho$.}
We assume here that $f : \{-1, 1\}^n \to \{-1, 1\}$ is monotone (but not necessarily unbiased). Let $\rho = 1 - 2q$. Recall that 
\[
\Stab_\rho[f] - \Stab_\rho[\Maj_n] = \sum_{k=1}^n q^k(1-q)^{n-k} \sum_{S \in \binom{[n]}{k}} G_S(f),
\]
where
\begin{align*}
G_S(f)
&= \Ex_{\bx} \left[f(\bx)f(\bx^{\oplus S}) - \Maj_n(\bx)\Maj_n(\bx^{\oplus S}) \right] \\
&= \Ex_{\bx} \left[ |\Maj_n(\bx)-\Maj_n(\bx^{\oplus S})| - |f(\bx)-f(\bx^{\oplus S})| \right].
\end{align*}
Here, we are using that sampling $(\bx, \by) \sim \DSBS_\rho^{\otimes n}$ is equivalent to sampling $\bx$ uniformly, a subset $S$ by including each element with probability $\frac{1 - \rho}{2} = q$ and setting $\by = \bx^{\oplus S}$ (flipping all coordinates in $S$).
We wish to apply \Cref{lem:gap} to this expansion. We have the following bounds. 
\begin{itemize}
\item For $|S| = 1$, let $S = \{i\}$. We have $G_{\{i\}}(f) = \Ex_{\bx} [2x_i(\Maj_n(\bx) - f(\bx))]$, since for any monotone function $h : \{-1, 1\}^n \to \{-1, 1\}$, it holds that $\Ex_{\bx} [|h(\bx) - h(\bx^{\oplus i})|] = \Ex_{\bx} [2 x_i h(\bx)]$. Thus $G_{\{i\}}[f]$ satisfies \Cref{lem:gap}(i) with $c=2$.
\item For $|S|= 2$, let $S = \{i, j\}$; we have two cases.  If $x_i = x_j$, then
\begin{align*}
&\frac{1}{2} \cdot \Ex_{\bx}\Brack{|\Maj_n(\bx) - \Maj_n(\bx^{\oplus S})| - |f(\bx) - f(\bx^{\oplus S})| ~\middle|~ x_i = x_j} \\
&= \frac{1}{4} \cdot \Ex_{\bx}\Brack{(x_i + x_j)\Paren{\Maj_n(\bx) - \Maj_n(\bx^{\oplus S})} - (x_i + x_j)\Paren{f(\bx) - f(\bx^{\oplus S})} ~\middle|~ x_i = x_j} \\
&= \frac{1}{2} \Ex_{\bx}\Brack{(x_i + x_j)\cdot (\Maj_n(\bx) - f(\bx))~\middle|~ x_i = x_j},
\end{align*}
where the first step uses the monotonicity of $\Maj_n$ and $f$.
If $x_i \ne x_j$, notice that $\Maj_n(\bx) = \Maj_n(\bx^{\oplus S})$. Thus, we have
\begin{align*}
&\frac{1}{2} \cdot \Ex_{\bx}\Brack{|\Maj_n(\bx) - \Maj_n(\bx^{\oplus S})| - |f(\bx) - f(\bx^{\oplus S})| ~\middle|~ x_i \ne x_j} \\
&\geq - \Pr_{\bx}\Brack{\bx \in \diffset_f \lor \bx^{\oplus S} \in \diffset_f  ~\middle|~ x_i \ne x_j} \\
& = -2 \Pr_{\bx}\Brack{\bx \in \diffset_f ~\middle|~ x_i \ne x_j}.
\end{align*}
Putting these bounds together, 
\[
G_{\{i, j\}}(f) \geq 
\frac{1}{2} \Ex_{\bx}\Brack{(x_i + x_j)\cdot (\Maj_n(\bx) - f(\bx)) ~\middle|~ x_i = x_j} - 2 \Pr_{\bx}\Brack{\bx \in \diffset_f ~\middle|~ x_i \ne x_j}.
\]
This satisfies \Cref{lem:gap}(ii) also with $c=2$.
\item For $|S| \ge 3$, we have 
\begin{align*}
G_S(f) & = \Ex_{\bx}\Brack{|\Maj_n(\bx) - \Maj_n(\bx^{\oplus S})| - |f(\bx) - f(\bx^{\oplus S})|} \\
&\geq - 2 \cdot \Pr_{\bx}\Brack{\bx \in \diffset_f \lor \bx^{\oplus S} \in \diffset_f}
\geq -4 \cdot  \Pr\Brack{\bx \in \diffset_f},
\end{align*}
This satisfies \Cref{lem:gap}(iii) also with $c=2$. 
\end{itemize}
From \Cref{lem:gap}, we conclude  $\Stab_\rho[f] - \Stab_\rho[\Maj_n] > 0$ as desired.

\paragraph{\boldmath Derivation for $\Phi_p$.}
We assume here that $f$ is unbiased.
Let $g^\star$ be the optimal function in the
definition of $\Phi_p[f]$ defined as\footnote{If the inner expectation $\Ex_{\bx_{\oS} \in \{-1, 1\}^{\oS}}[f(\bx_S \circ \bx_{\oS})]$ is zero, then we may select $g^\star(\bx_S \circ \bzero_{\oS})$ arbitrarily.} $g^\star(\bx_S \circ \bzero_{\oS}) := \sign\left(\Ex_{\bx_{\oS} \in \{-1, 1\}^{\oS}}[f(\bx_S \circ \bx_{\oS})]\right)$.  Note that $\Phi_p[\Maj_n] \geq \Ex_{\bx, S}[\Maj_n(\bx) \cdot g^\star\Paren{\bx_S \circ \bzero_{\oS}}]$.  Therefore, 
\begin{align*}
\Phi_{p}[\Maj_n] - \Phi_{p}[f] 
& \geq \Ex_{\bx, S}\Brack{ \Paren{ \Maj_n(\bx) - f(\bx) } \cdot g^\star\Paren{\bx_S \circ \bzero_{\oS}} } 
= \sum_{k = 1}^n q^k (1 - q)^{n-k} 
\sum_{S \in \binom{n}{k}} G_S(f),
\end{align*}
where $q = p$ and 
\[
G_S(f) = \Ex_{\bx}\Brack{ \Paren{ \Maj_n(\bx) - f(\bx) } \cdot g^\star\Paren{\bx_S \circ \bzero_{\bar{\bS}}}},
\]
and we note that the $k = 0$ term disappears since $f, \Maj_n$ are both unbiased.

As noted in \citep{ivanisvili2025nicd}, we may assume\footnote{This follows directly from monotonicity. Even without monotonicity, this simply follows by ``flipping coordinate $i$'' that violates this condition, which does not change $\Phi_p[f]$; see the proof of \Cref{lem:three-case-reduction} for a formalization.} w.l.o.g. that the first-level Fourier coefficient is non-negative, i.e., 
$\Ex_{\bx} [x_i \cdot f(\bx)] \ge 0$. Observe that this implies $g^\star(x_i \circ \bzero_{[n] \setminus \{i\}}) = x_i$.
\begin{lemma}
For any event $C$, 
$\Ex_{\bx}[(\Maj_n(\bx) - f(\bx)) \cdot g^\star(\bx_S \circ \bzero_{\oS}) \mid C] \geq -2 \Pr[\bx \in E_f \mid C].$
\label{lem:small}
\end{lemma}
\begin{proof}
Note that since $|g^\star| \leq 1$, we have that $(\Maj_n(x) - f(x)) \cdot g^\star(\cdot)$ evaluates to $0$ if $x \notin E_f$, and is lower bounded by $-2$ if $x \in E_f$.  The proof follows. 
\end{proof}

We wish to apply \Cref{lem:gap} to
this expansion.  First, consider the case of $|S| = 1$, and let $S = \{i\}$.
Since 
$g^\star(x_i \circ \bzero_{n \setminus \{i\}}) = x_i$, we have $G_{\{i\}}(f) = \Ex_{\bx}[x_i(\Maj_n(\bx) - f(\bx))]$; this satisfies \Cref{lem:gap}(i) with $c=1$.

For $|S| > 1$, we will separate into two cases based on whether $f$ is monotone.\\

\noindent {\em When $f$ is monotone.}
In this case, we have that

\begin{itemize}

\item For $|S|= 2$, let $S = \{i,j\}$.
We have two cases.  If $x_i = x_j$, for an unbiased monotone function,  notice that we simply have $g^\star(\bx_S \circ \bzero_{\oS}) = \sign(x_i + x_j) = \frac{x_i + x_j}{2}$. Thus,
\begin{align*}
& \frac{1}{2} \Ex_{\bx}\Brack{ \Paren{ \Maj_n(\bx) - f(\bx) } \cdot g^\star\Paren{\bx_{\{i,j\}} \circ \bzero_{[n] \setminus \{i,j\}}} \mid x_i = x_j}\\
& = \frac{1}{4}
\Ex_{\bx}\Brack{ (x_i + x_j) \Paren{ \Maj_n(\bx) - f(\bx) } \mid x_i = x_j}.
\end{align*}
If $x_i \neq x_j$, applying \Cref{lem:small}, we have 
\[
\frac{1}{2} \Ex_{\bx}\Brack{ \Paren{ \Maj_n(\bx) - f(\bx) } \cdot g^\star\Paren{\bx_{\{i,j\}} \circ \bzero_{[n] \setminus \{i,j\}}} \mid x_i \neq x_j}
\geq 
-\Pr[\bx \in E_f \mid x_i \neq x_j].
\]
Putting these bounds together, we have \Cref{lem:gap}(ii) with $c = 1$. 
\item For $|S| \ge 3$, we have by \Cref{lem:small} that
$G_S(f) \geq -2 \Pr[\bx \in E_f]$.  Thus, we have \Cref{lem:gap}(iii) with $c = 1$. 
\end{itemize}
From \Cref{lem:gap}, we conclude  $\Phi_p[\Maj_n] - \Phi_p[f] > 0$ as desired.\\

\noindent {\em When $f$ is not monotone.}
Finally, consider the case where $f$ is non-monotone. In this case, we simply apply \Cref{lem:small} to
obtain a tail bound on all terms, including the second term, 
\begin{align*}
\sum_{k=2}^{n} q^k (1-q)^{n-k} \sum_{S \in \binom{[n]}{k}} G_S(f) \geq
\sum_{k=2}^{n} q^k (1-q)^{n-k} \sum_{S \in \binom{[n]}{k}} (-2\mu_f) \geq -2\mu_f \cdot q^2 \binom{n}{2}
= -\mu_f \cdot q^2 n(n - 1).
\end{align*}
Combining with the first term, we get
\begin{align*}
\Phi_{p}[\Maj_n] - \Phi_{p}[f] &\geq q(1 - q)^{n - 1} \cdot 2\mu_f - \mu_f \cdot q^2 n(n - 1) 
= \mu_f q (2(1 - q)^{n - 1} - qn(n - 1)) \\
&\geq \mu_f q (2(1 - (n - 1)q) - qn(n - 1))) > 0,
\end{align*}
where we used Bernoulli's inequality in the penultimate step and our assumption $0 < q < \frac{2}{(n+2)(n-1)}$ in the last step.

\section{\boldmath The Case of \texorpdfstring{$n=3$}{n=3} (Proof of \texorpdfstring{\Cref{thm:n3}}{Theorem~\ref{thm:n3}})}

While the previous sections established  bounds for general $n$, the case of $n=3$ is small enough that we can evaluate the veracity of both
conjectures.  In fact, we use techniques from \Cref{sec:tight}.

\paragraph{\boldmath Derivation for $\Stab_\rho$.}
Assume that $f$ is monotone and odd.
We can use the same expression for $G_S(f)$ for the case $k = 1, 2$ as in \Cref{sec:tight}. Finally, notice that for $k = 3$ (i.e., $S = \{1, 2, 3\}$, we simply have $G_S(f) = 0$ since $f$ is odd. As a result, following the inequalities in \Cref{lem:gap} but without the $k \geq 3$ term, we get
\[
G_q[f] \geq c \cdot \mu_f \cdot q \left( 2(1-q)^{n-1} - q(1-q)^{n-2}\frac{(n-1)^2}{2}\right) 
= c \cdot \mu_f \cdot q (1 - q) \left(2 - 4q\right) > 0,
\]
where in the first equality we use $n = 3$ and in the last inequality we use $0 < q < 1/2$.

\paragraph{\boldmath Derivation for $\Phi_p$.} In this case, we will use the following lemma.

\begin{lemma} \label{lem:three-case-reduction}
For any unbiased $f: \{-1, 1\}^3 \to \{-1, 1\}$, there exists another unbiased function $f': \{-1, 1\}^3 \to \{-1, 1\}$ such that $\Phi_{\rho}[f'] = \Phi_\rho[f]$ and there is an optimum $g^\star$ in the definition of $\Phi_{\rho}[f]$ that satisfies:
\begin{itemize}
\item[(i)] $g^\star(x_i \circ \bzero_{[n] \setminus \{i\}}) = x_i$ for all $i \in [3]$ and $\bx \in \{-1, 1\}^3$ and
\item[(ii)] $g^\star(\bx_S \circ \bzero_{\oS}) = (x_i + x_j)/2$ for all $i \ne j$ and $\bx \in \{-1, 1\}^3$ such that $x_i = x_j$.
\end{itemize}
\end{lemma}

Before we prove \Cref{lem:three-case-reduction}, we show how to use this to finish the proof. From the lemma, we may assume w.l.o.g. that $f$ itself has $g^\star$ that satisfies the two properties. Note that these are the only two properties used in derivation of the monotone case of $\Phi_\rho[f]$ in \Cref{sec:tight}. Again, similar to the derivation for $\Stab$ above, by recognizing that the $k = 3$ term cancels out and following the inequalities in \Cref{lem:gap} but without the $k \geq 3$ term, we get
\begin{align*}
G_q[f] &\geq c \cdot \mu_f \cdot q \left( 2(1-q)^{n-1} - q(1-q)^{n-2}\frac{(n-1)^2}{2}\right) \geq 0.
\end{align*}

We will now prove \Cref{lem:three-case-reduction}.

\begin{proof}[Proof of \Cref{lem:three-case-reduction}]
We pick $f'$ of the form $f' = f(\bx \oplus a)$ for some $a \in \{-1, 1\}^3$ that maximizes $\Ex_{\bx}[x_i \cdot f'(\bx)]$, ties  broken arbitrarily. Note that this ensures that $\Phi_{\rho}[f'] = \Phi_\rho[f]$. Furthermore, our choice of $f'$ ensures that $\Ex_{\bx}[x_i \cdot f'(\bx)] \geq 0$ for all $i \in [n]$. That is, property (i) holds.

Now, suppose that property (ii) does not hold. This means that there exists $\bx \in \{-1, 1\}^3$ and $i \ne j$ such that $x_i = x_j$ with $\Ex_{\bx_{\oS} \in \{-1, 1\}^{\oS}}[f'(\bx_S \circ \bx_{\oS})] \cdot (x_i + x_j) < 0$.  By symmetry, we may assume that $i = 1, j = 2$, and $x_i = x_j = 1$. This implies that $f'(1, 1, 1) = f'(1, 1, -1) = -1$. 

However, since $\Ex_{\bx}[x_i \cdot f'(\bx)] \geq 0$ and $f'$ is unbiased, plugging in $i = 1$ and $i = 2$, respectively, implies that
\[
f'(1, -1, -1) = g(1, -1, 1) = 1 \text{ and } f'(-1, 1, 1) = f'(-1, 1, -1) = 1,
\]
respectively. Finally, since $f'$ is unbiased, we also have $f'(-1, -1, 1) = f'(-1, -1, -1) = -1$. This means that $f'(\bx) = -x_1x_2$. However, in this case, we may take $f'(\bx) = x_1x_2$ instead, which satisfies both properties (i) and (ii).
\end{proof}

\bibliographystyle{plainnat}
\bibliography{main.bbl}

\end{document}